\documentclass[10pt,reqno,a4paper]{amsart}

\usepackage{mathrsfs}
\usepackage{dsfont}   

\usepackage{geometry}
\newtheorem{proposition}{Proposition}[section]
\newtheorem{corollary}[proposition]{Corollary}
\newtheorem{lemma}[proposition]{Lemma}
\newtheorem{theorem}[proposition]{Theorem}
\theoremstyle{definition}
\newtheorem{definition}[proposition]{Definition}

\newcommand{\I}{\mathds{1}}
\newcommand{\id}{\mathrm{id}}
\newcommand{\is}[2]{{\left\langle{#1}\,\vline\,#2\right\rangle}}
\newcommand{\ket}[1]{{\left|#1\right\rangle}}
\newcommand{\bra}[1]{{\left\langle#1\right|}}
\newcommand{\ketbra}[2]{\ket{#1}\!\bra{#2}}
\newcommand{\ph}{\varphi}
\newcommand{\tens}{\otimes}
\newcommand{\bigstrut}{\bigl.\bigr.}

\newcommand{\CC}{\mathbb{C}}

\newcommand{\sT}{\mathsf{T}}

\newcommand{\cH}{\mathscr{H}}
\newcommand{\cK}{\mathscr{K}}
\newcommand{\cN}{\mathcal{N}}

\DeclareMathOperator{\B}{B}
\DeclareMathOperator{\Mat}{\mathsf{Mat}}
\DeclareMathOperator{\Tr}{Tr}

\newcommand{\stCn}[1]{e_{#1}^{\scriptscriptstyle{n}}}
\newcommand{\stCk}[1]{e_{#1}^{\scriptscriptstyle{k}}}

\title{Invariants of the quantum graph of the partial trace}

\author{Wojciech Paupa}
\address{Department of Mathematical Methods in Physics, Faculty of Physics, University of Warsaw}
\email{w.paupa@student.uw.edu.pl}

\author{Piotr M.~Sołtan}
\address{Department of Mathematical Methods in Physics, Faculty of Physics, University of Warsaw}
\email{piotr.soltan@fuw.edu.pl}

\makeatletter
\@namedef{subjclassname@2020}{\textup{2020} Mathematics Subject Classification}
\makeatother


\begin{document}

\begin{abstract}
We compute the independence number, zero-error capacity, and the values of the Lov\'asz function and the quantum Lov\'asz function for the quantum graph associated to the partial trace quantum channel $\operatorname{Tr}_n\otimes\mathrm{id}_k\colon\operatorname{B}(\mathbb{C}^n\otimes\mathbb{C}^k)\to\operatorname{B}(\mathbb{C}^k)$.
\end{abstract}

\maketitle

\section{Introduction}

In the paper \cite{DuanSeveriniWinter2013}
the authors introduced a novel approach to the study of quantum channels, i.e.~completely positive trace preserving maps between spaces of operators on finite dimensional Hilbert spaces. This new approach is based on the very appropriate analogy with classical communication channels first studied by C.~Shannon in \cite{Shannon1948}, where to each \emph{classical channel} (a mapping from a finite set into probability measures on another finite set) a finite graph called the \emph{confusability graph} is assigned. This graph carries important information about the channel, the most important bit being its \emph{Shannon capacity} which is expressible in terms of the so called \emph{independence number} (maximal cardinality of a subset of vertices without any connections) of the graph and its strong powers. For more details we refer to \cite{Shannon1948} with further developments described e.g.~in \cite{Lovasz1979}.

Now let $\cH$ and $\cK$ be finite dimensional Hilbert spaces and $\cN\colon\B(\cH)\to\B(\cK)$ be a quantum channel. It is a well known consequence of the Stinespring representation theorem (\cite{Stinespring1955}) that any such map is of the form
\[
\cN(A)=\sum_{i=1}^rK_iAK_i^*,\qquad{A}\in\B(\cH),
\]
where $K_1,\dotsc,K_r\in\B(\cH,\cK)$ are so called \emph{Kraus operators} (see \cite[Section 3]{Kraus1983}). The authors of \cite{DuanSeveriniWinter2013} show that the linear subspace
\[
S_\cN=\operatorname{span}\{K_i^*K_j\,|\,i,j=1,\dotsc,r\}\subset\B(\cH)
\]
assigned to the quantum channel $\cN$ can be regarded as a non-commutative (quantum) analog of what the confusability graph is for a classical channel.

The first observation one needs to make is that $S_\cN$ does not depend on the choice of the Kraus operators (see \cite[Lemma 1]{DuanSeveriniWinter2013}). Since the fact that $\cN$ preserves the trace is equivalent to the condition $\sum\limits_iK_i^*K_i=\I$, we see that $S_\cN$ contains the identity operator $\I\in\B(\cH)$, and by its very construction, $S_\cN$ is closed with respect to the operation of taking Hermitian adjoints $\B(\cH)\ni{A}\mapsto{A^*}\in\B(\cH)$. Subspaces of $\B(\cH)$ with these two properties are called \emph{operator systems}. They have been studied for a long time (see e.g.~\cite{ChoiEffros1977,Paulsen2002}) and have become integrated into the theory of operator algebras.

Now let us indicate how $S_\cN$ is used to obtain information about the channel $\cN$. The analogy between $S_\cN$ and the confusability graph of a classical channel is very well described in \cite{DuanSeveriniWinter2013}, so we will restrict ourselves only to a very brief explanation. Let $N$ be a classical channel from $X$ to $Y$, i.e.~a map from $X$ to probability measures on a set $Y$ with the measure corresponding to $x$ denoted by $N(\cdot|x)$. Then the confusability graph $G_N$ has vertices $X$ and $x,x'\in{X}$ are connected by an edge if and only if the probability measures $N(\cdot|x)$ and $N(\cdot|x')$ have overlapping supports (this, in particular, implies that each vertex is connected to itself). Thus, for example, the independence number of $G_N$ provides information about how much information can be transmitted via the channel without errors. It turns out that one can extend the channel $N$ to a quantum channel $\cN\colon\B(\CC^{|X|})\to\B(\CC^{|Y|})$ by defining its Kraus operators $\{K_{x,y}\,|\,x\in{X},\:y\in{Y}\}$ as
\[
K_{x,y}=\sqrt{N(y|x)}\,\ketbra{y}{x},
\]
where $\bigl\{\ket{x}\,\bigr|\bigl.\,x\in{X}\bigr\}$ and $\bigl\{\ket{y}\,\bigr|\bigl.\,y\in{Y}\bigr\}$ are the standard bases of $\CC^{|X|}$ and $\CC^{|Y|}$. One can check that the corresponding subspace $S_\cN$ is
\[
\operatorname{span}\bigl\{\ketbra{x}{x'}\,\bigr|\bigl.\,\text{$x$ and $x'$ are connected by an edge in $G_N$}\bigr\}.
\]
Thus $S_\cN$ contains the information carried by the adjacency matrix of $G_N$ (in particular the graph $G_N$ is recoverable from $S_\cN$). It is for this reason that for any quantum channel $\cN\colon\B(\cH)\to\B(\cK)$ the object $S_\cN$ is called the \emph{quantum confusability graph} of $\cN$. More generally we will say that an operator system $S\subset\B(\cH)$ describes a \emph{quantum graph}.

Let us remark that the theory of quantum graphs has in recent years undergone very substantial development and is by now a field of research on its own. The papers of N.~Weaver (\cite{Weaver2012,Weaver2021}) as well as the work of many other authors \cite{MustoReutterVerdon2019,BrannanChirvasituEiflerHarrisPaulsenSuWasilewski2020,ChirvasituWasilewski2022,Daws2024} have provided multiple points of view on the theory and opened new avenues of further research and applications.

Returning to the quantum graph $S_\cN$ of a quantum channel $\cN$, we will be interested in four numerical invariants of $\cN$ introduced in \cite{DuanSeveriniWinter2013}. They are
\begin{itemize}
\item the independence set,
\item zero-error capacity,
\item value of the Lov\'asz function $\vartheta$,
\item value of the quantum Lov\'asz function $\tilde{\vartheta}$.
\end{itemize}
All these are described and studied in \cite{DuanSeveriniWinter2013}, but we will briefly recall their definitions. Before doing this, however, we must recall the notion of a strong product of quantum graphs. If $S_1\subset\B(\cH_1)$ and $S_2\in\B(\cH_2)$ are operator systems describing two quantum graphs, the strong product of these quantum graphs is described by the tensor product $S_1\tens{S_2}\subset\B(\cH_1)\tens\B(\cH_2)=\B(\cH_1\tens\cH_2)$. In case $S_1$ and $S_2$ are quantum graphs arising from quantum channels $\Phi_1\colon\B(\cH_1)\to\B(\cK_1)$ and $\Phi_2\colon\B(\cH_2)\to\B(\cK_2)$ then it is easy to see that the operator system $S_1\tens{S_2}$ corresponds to the quantum channel $\Phi_1\tens\Phi_2\colon\B(\cH_1\tens\cH_2)\to\B(\cK_1\tens\cK_2)$.

\begin{definition}
Let $\cH$ be a finite dimensional Hilbert space and $S\subset\B(\cH)$ an operator system.
\begin{enumerate}
\item The \emph{independence number} $\alpha(S)$ of $S$ is the maximal cardinality of an orthonormal system $\{\eta_1,\dotsc,\eta_N\}$ such that $\ketbra{\eta_p}{\eta_q}\in{S^\perp}$ for all $p,q\in\{1,\dotsc,N\}$ with $p\neq{q}$ (with $S^\perp$ denoting the orthogonal complement of $S$ with respect to the Hilbert-Schmidt scalar product).
\item The \emph{zero-error capacity} of $S$ is $C_0(S)=\lim\limits_{m\to\infty}\tfrac{1}{m}\log\bigl(\alpha(S^{\tens{m}})\bigr)$, where $S^{\tens{m}}$ is the strong power of the quantum graph $S$ (the limit exists, since $\alpha$ is a sub-multiplicative function).
\item The \emph{Lov\'asz function} $\vartheta$ assigns to $S$ the value
\[
\max\bigl\{\|\I+T\|\,\bigr|\bigl.\,T\in{S^\perp},\:\I+T\geq{0}\bigr\}
\]
(with $S^\perp$ again denoting the orthogonal complement of $S$ with respect to the Hilbert-Schmidt scalar product on $\B(\cH)$, and $\|\cdot\|$ denoting the operator norm, see Section \ref{sect:not} below).
\item The \emph{quantum Lov\'asz function} $\tilde{\vartheta}$ assigns to $S$ the value
\[
\sup_m\vartheta\bigl(\B(\CC^m)\tens{S}\bigr).
\]
\end{enumerate}
\end{definition}
Let us note that the orthogonal complement of $S\tens\B(\CC^m)\subset\B(\cH\tens\CC^m)$ with respect to the Hilbert-Schmidt scalar product is $S^\perp\tens\B(\CC^m)$.

We will compute the above mentioned four values for one of the most commonly used quantum channels, the partial trace $\Tr_n\tens\id_k\colon\B(\CC^n)\tens\B(\CC^k)\to\B(\CC^k)$ understood as a map $\B(\CC^n\tens\CC^k)\to\B(\CC^k)$ (see Section \ref{sect:not} for details of our notation). The computations are presented in Section \ref{sect:comp}.

\subsection{Notational conventions}\label{sect:not}

All Hilbert spaces will be over the field of complex numbers and all scalar products $\is{\cdot}{\cdot}$ will be linear in the second variable. For a Hilbert space $\cH$ and a vector $\psi\in\cH$ the symbol $\ket{\psi}$ will denote two different objects: the vector $\psi$ itself and the unique linear mapping $\CC\to\cH$ taking $1$ to $\psi$. It is clear that these can be identified. The Hermitian adjoint of $\ket{\psi}$ will be denoted by the symbol $\bra{\psi}$: it is the linear map $\cH\to\CC$ taking $\ph\in\cH$ to the scalar $\is{\psi}{\ph}$. Thus for $\xi\in\cH$ and $\eta\in\cK$ the operator $\ketbra{\xi}{\eta}\in\B(\cK,\cH)$ is the composition of $\bra{\eta}\colon\cK\to\CC$ and $\ket{\xi}\colon\CC\to\cH$.

The canonical trace on $\B(\cH)$ will be denoted by $\Tr$. The corresponding Hilbert-Schmidt scalar product on $\B(\cH)$ is then defined as
\[
\is{A}{B}=\Tr(A^*B),\qquad{A,B}\in\B(\cH).
\]
Note that the cyclicity of the trace, i.e.~the property that $\Tr(AB)=\Tr(BA)$, holds also if $A\in\B(\cK,\cH)$, $B\in\B(\cH,\cK)$ and the traces on the two sides of the equation are in fact traces on different spaces of operators.

The norm on $\B(\cH,\cK)$ will be denoted by $\|\cdot\|$. This is always the operator norm defined by
\[
\|A\|=\sup_{\|\psi\|=1}\|A\psi\|.
\]

The symbol $\tens$ will, depending on the context, denote both the tensor product of Hilbert spaces as well as the tensor product of spaces of operators. We will frequently use the natural identification $\B(\cH)\tens\B(\cK)=\B(\cH\tens\cK)$ for finite dimensional Hilbert spaces $\cH$ and $\cK$.

In order to avoid any confusion, in what follows we will use the symbols $\Tr_n$ and $\Tr_k$ to denote traces on $\B(\CC^n)$ and $\B(\CC^k)$ respectively. Similarly $\id_n$ and $\id_k$ will be the identity maps on these spaces and $\sT_k$ will denote the transposition map on $\B(\CC^k)$. This last map requires a choice of basis (or any other equivalent way to identify the dual space of $\CC^k$ with itself) and we are choosing the standard basis of $\CC^k$.\footnote{For any Hilbert space $\cH$ the \emph{transposition} is the mapping $\B(\cH)\to\B(\cH^*)$ taking $A\in\B(\cH)$ to the operator $\bra{\psi}\mapsto\bra{A^*\psi}$ on $\cH^*$. Upon some \emph{linear} identification of $\cH^*$ with $\cH$ we can regard it as a map $\B(\cH)\to\B(\cH)$.} With this choice $\sT_k$ becomes the usual transposition of matrices if we identify $\B(\CC^k)$ with scalar $k\times{k}$ matrices via the standard basis of $\CC^k$. In particular $\sT_k$ is an involutive map.

Lastly $\I_n$ and $\I_k$ will denote the identity operators on $\CC^n$ and $\CC^k$ respectively while the unit of $\B(\CC^n\tens\CC^k)$ will be denoted simply by $\I$.

\section{Computation of the invariants}\label{sect:comp}

Let $\{\stCn{1},\dotsc,\stCn{n}\}$ be the standard basis of $\CC^n$ and $\{\stCk{1},\dotsc,\stCk{k}\}$ the standard basis of $\CC^k$. The quantum channel $\Tr_n\tens\id_k\colon\B(\CC^n\tens\CC^k)\to\B(\CC^k)$ has the following Kraus form:
\[
(\Tr_n\tens\id_k)(T)=\sum_{i=1}^nP_iTP_i^*,
\]
where $P_i=\bra{\stCn{i}}\tens\I_k\in\B(\CC^n\tens\CC^k,\CC^k)$ ($i=1,\dotsc,n$). The corresponding operator system is
\[
S=\operatorname{span}\bigl\{P_i^*P_j\,\bigr|\bigl.\,i,j=1,\dotsc,n\bigr\}
=\operatorname{span}\bigl\{\ketbra{\stCn{i}}{\stCn{j}}\tens\I_k\,\bigr|\bigl.\,i,j=1,\dotsc,n\bigr\}=\B(\CC^n)\tens\I_k
\]
understood as a subset of $\B(\CC^n)\tens\B(\CC^k)=\B(\CC^n\tens\CC^k)$.

Let us also describe the orthogonal complement $S^\perp$ of $S$ with respect to the Hilbert-Schmidt scalar product on $\B(\CC^n\tens\CC^k)$. It follows easily from the identity $\Tr_n\tens\Tr_k=\Tr$ (trace on $\B(\CC^n\tens\CC^k)$) that for any $A\in\B(\CC^n)$ and $X\in\B(\CC^n\tens\CC^k)$ we have
\[
\Tr\bigl((A\tens\I_k)X\bigr)=\Tr_n\Bigl(A\bigl((\id_n\tens\Tr_k)(X)\bigr)\Bigr).
\]
Consequently $X\in{S^\perp}$ is equivalent to $(\id_n\tens\Tr_k)(X)$ being orthogonal to all elements of $\B(\CC^n)$, i.e.~to $(\id_n\tens\Tr_k)(X)=0$. It follows that
\begin{equation}\label{eq:Sperp}
S^\perp=\bigl\{X\in\B(\CC^n\tens\CC^k)\,\bigr|\bigl.\,(\id_n\tens\Tr_k)(X)=0\bigr\}.
\end{equation}

\subsection{The independence number}

\begin{proposition}\label{prop:alphaS}
We have $\alpha(S)=k$.
\end{proposition}

\begin{proof}
First let us note the canonical isomorphism $\CC^n\tens\CC^k\cong\bigoplus\limits_{i=1}^n\CC^k$ which identifies $\psi\in\CC^n\tens\CC^k$ with
\[
\begin{bmatrix}
P_1\psi\\[-5pt]\vdots\\[-3pt]P_n\xi
\end{bmatrix}.
\]
For any $\xi,\eta\in\CC^n\tens\CC^k$ we have
\begin{align*}
\Tr\Bigl(\bigl(\ketbra{\xi}{\eta}\bigr)^*\bigl(\ketbra{\stCn{i}}{\stCn{j}}\tens\I_k\bigr)\Bigr)
&=\Tr\Bigl(\ketbra{\eta}{\xi}\bigl(\ket{\stCn{i}}\tens\I_k\bigr)\bigl(\bra{\stCn{j}}\tens\I_k\bigr)\Bigr)
=\Tr\bigl(\ketbra{\eta}{\xi}P_i^*P_j\bigr)\\
&=\Tr\bigl(P_j\ketbra{\eta}{\xi}P_i^*\bigr)
=\Tr\bigl(\ketbra{P_j\eta}{P_i\xi}\bigr)=\is{P_i\xi}{P_j\eta}.
\end{align*}
Now if $\{\eta_1,\dotsc,\eta_N\}$ is an orthonormal system such that for any $p\neq{q}$ we have $\ketbra{\eta_p}{\eta_q}\in{S^\perp}$ (in other words $\Tr\Bigl(\bigl(\ketbra{\eta_p}{\eta_q}\bigr)^*\bigl(\ketbra{\stCn{i}}{\stCn{j}}\tens\I_k\bigr)\Bigr)=0$ for all $i,j$) then the vectors
\[
\bigl\{P_i\eta_p\,\bigr|\bigl.\,p=1,\dotsc,N,\:i=1,\dotsc,n\bigr\}
\]
are pairwise orthogonal. Moreover, since all $\eta_p$'s are non-zero, for each $p$ there exists $i(p)\in\{1,\dotsc,n\}$ such that $P_{i(p)}\eta_p\neq{0}$. In particular $\{P_{i(1)}\eta_1,\dotsc,P_{i(N)}\eta_N\}$ is a system of non-zero pairwise orthogonal vectors in $\CC^k$. This means that $N\leq{k}$ and consequently $\alpha(S)\leq{k}$.

On the other hand taking $N=k$ and putting
\[
\eta_p=\begin{bmatrix}
0\\[-6pt]\vdots\\[-4pt]\stCk{p}\\[-3pt]\vdots\\[-2pt]0
\end{bmatrix},\qquad{p}=1\dotsc,k
\]
($p$-th vector of the standard basis in $p$-th component, remaining ones equal zero). we obtain an orthonormal system $\{\eta_1,\dotsc,\eta_k\}$ in $\bigoplus\limits_{i=1}^n\CC^k\cong\CC^n\tens\CC^k$ such that for any distinct $p,q$ we have $\ketbra{\eta_p}{\eta_q}\in{S^\perp}$.
\end{proof}

\subsection{Zero-error capacity}

The next computation is a very simple consequence of Proposition \ref{prop:alphaS}.

\begin{corollary}
We have $C_0(S)=\log{k}$.
\end{corollary}

\begin{proof}
We first note that $S^{\tens{m}}=(\B(\CC^n)\tens\I_k)^{\tens{m}}=\B(\CC^n)^{\tens{m}}\tens\I_{k^m}=\B\bigl(\CC^{n^m}\bigr)\tens\I_{k^m}$.

Now from Proposition \ref{prop:alphaS} we know that $\alpha(S^{\tens{m}})=k^m$ and consequently $\tfrac{1}{m}\log\bigl(\alpha(S^{\tens{m}})\bigr)=\log{k}$ for all $m$.
\end{proof}

\subsection{The Lov\'asz function and the quantum Lov\'asz function}

\begin{lemma}\label{lem:cbnormk}
Let $X\in\B(\CC^n\tens\CC^k)$ be positive. Then $\|X\|\leq{k}\|(\Tr_n\tens\id_k)(X)\|$.
\end{lemma}

\begin{proof}
We use D.~Choi's inequality from \cite[Theorem 2]{Choi2017}, i.e.~the fact that
\[
\I_n\tens\bigl((\Tr_n\tens\id_k)(X)\bigr)\geq(\id_n\tens\sT_k)(X).
\]
It follows that
\[
\bigl\|(\Tr_n\tens\id_k)(X)\bigr\|=\Bigl\|\I_n\tens\bigl((\Tr_n\tens\id_k)(X)\bigr)\Bigr\|\geq\bigl\|(\id_n\tens\sT_k)(X)\bigr\|.
\]
Now we note that $\|X\|=\bigl\|(\id_n\tens\sT_k)(\id_n\tens\sT_k)(X)\bigr\|\leq\|\id_n\tens\sT_k\|\bigl\|(\id_n\tens\sT_k)(X)\bigr\|$. The norm of the operator $\id_n\tens\sT_k\colon\B(\CC^n)\tens\B(\CC^k)\to\B(\CC^n)\tens\B(\CC^k)$ is majorized by (and equal to if $n\geq{k}$) the \emph{completely bounded norm} $\|\sT_k\|_{\text{\tiny{\rm{cb}}}}$ (\cite[Chapter 1]{Paulsen2002}) which is equal to $k$ by the main result of \cite{Tomiyama1983}.
\end{proof}

\begin{theorem}\label{thm:theta}
We have
\[
\vartheta(S)=\begin{cases}
k^2&n>k\\
nk&n\leq{k}
\end{cases}.
\]
\end{theorem}

\begin{proof}
Take $T\in{S^\perp}$ such that $\I+T\geq{0}$. Then by \eqref{eq:Sperp} we have
\begin{align*}
\|\I+T\|\leq\Tr(\I+T)&=\Tr_n\bigl((\id_n\tens\Tr_k)(\I+T)\bigr)\\
&=\Tr_n\bigl((\id_n\tens\Tr_k)(\I)+(\id_n\tens\Tr_k)(T)\bigr)\\
&=\Tr_n\bigl((\id_n\tens\Tr_k)(\I)\bigr)=\Tr_n(k\I_n)=nk.
\end{align*}
In case $n>k$ this estimate can be improved as follows: applying Lemma \ref{lem:cbnormk} to the positive operator $X=\I+T$ and again using \eqref{eq:Sperp} we obtain
\[
\|\I+T\|\leq{k}\bigl\|(\id_n\tens\Tr_k)(\I+T)\bigr\|=k\bigl\|(\id_n\tens\Tr_k)(\I)\bigr\|=k\|k\I_n\|=k^2
\]
which shows that $\vartheta(S)\leq{k^2}$ when $n>k$ and $\vartheta(S)\leq{nk}$ when $n\leq{k}$. We will prove the converse inequalities by exhibiting an examples of $T\in{S^\perp}$ such that $\I+T\geq{0}$ and $\|\I+T\|=k^2$ in the former case and $\|\I+T\|=nk$ in the latter.

Let us first consider the case $n\geq{k}$. Put
\[
\I+T=k\biggl(\sum_{i,j=1}^k\ketbra{\bigstrut\stCn{i}}{\stCn{j}}\tens\ketbra{\bigstrut\stCk{i}}{\stCk{j}}+\sum_{l=k+1}^n\ketbra{\stCn{l}}{\stCn{l}}\tens\ketbra{\stCk{k}}{\stCk{k}}\biggr).
\]
As an illustration let us identify $\B(\CC^n\tens\CC^k)=\B(\CC^n)\tens\B(\CC^k)$ with the space of block matrices $\Mat_k(\Mat_n(\CC))$ as follows: first using the standard bases of $\CC^n$ and $\CC^k$ we identify $\B(\CC^n)$ and $\B(\CC^k)$ with $\Mat_n(\CC)$ and $\Mat_k(\CC)$ respectively and then for $A\in\Mat_n(\CC)$ and $B\in\Mat_k(\CC)$ we identify $A\tens{B}$ with the block matrix
\[
\left[\begin{array}{c|c|c}
b_{1,1}A&\dotsm&b_{k,1}A\\
\hline
\vdots&\ddots&\vdots\\
\hline
b_{k,1}A&\dotsm&b_{k,k}A
\end{array}\right].
\]
With this identification, for $n=4$ and $k=3$ the element $\I+T\in\B(\CC^4\tens\CC^3)\cong\Mat_2(\Mat_4(\CC))$ is
\[
\left[\begin{array}{cccc|cccc|cccc}
3 & 0 & 0 & 0 & 0 & 3 & 0 & 0 & 0 & 0 & 3 & 0 \\
0 & 0 & 0 & 0 & 0 & 0 & 0 & 0 & 0 & 0 & 0 & 0 \\
0 & 0 & 0 & 0 & 0 & 0 & 0 & 0 & 0 & 0 & 0 & 0 \\
0 & 0 & 0 & 0 & 0 & 0 & 0 & 0 & 0 & 0 & 0 & 0 \\
\hline
0 & 0 & 0 & 0 & 0 & 0 & 0 & 0 & 0 & 0 & 0 & 0 \\
3 & 0 & 0 & 0 & 0 & 3 & 0 & 0 & 0 & 0 & 3 & 0 \\
0 & 0 & 0 & 0 & 0 & 0 & 0 & 0 & 0 & 0 & 0 & 0 \\
0 & 0 & 0 & 0 & 0 & 0 & 0 & 0 & 0 & 0 & 0 & 0 \\
\hline
0 & 0 & 0 & 0 & 0 & 0 & 0 & 0 & 0 & 0 & 0 & 0 \\
0 & 0 & 0 & 0 & 0 & 0 & 0 & 0 & 0 & 0 & 0 & 0 \\
3 & 0 & 0 & 0 & 0 & 3 & 0 & 0 & 0 & 0 & 3 & 0 \\
0 & 0 & 0 & 0 & 0 & 0 & 0 & 0 & 0 & 0 & 0 & 3
\end{array}\right].
\]
We note that upon re-ordering of the basis $\I+T$ will have the following matrix representation
\[
\left[\begin{array}{c|cc}
0_{(k-1)n} & 0 & 0 \\
\hline
0 & kJ_k & 0\\
0 & 0 & k\I_{n-k}
\end{array}\right],
\]
where $J_k$ is a $k\times{k}$ matrix all of whose entries are $1$ and $0_{(k-1)n}$ is the zero $(k-1)n\times(k-1)n$ matrix. This immediately shows that $\I+T\geq{0}$ and that $\|\I+T\|=\max\bigl\{\|kJ_k\|,\|k\I_{n-k}\|\bigr\}=k^2$ because $J_k$ is $k$ times the projection onto the subspace spanned by the vector vector $\sum\limits_{p=1}^k\ket{\stCk{p}}$.

Finally
\begin{align*}
(\id_n\tens\Tr_k)(\I+T)
=k\biggl(\sum_{i=1}^k\ketbra{\stCn{i}}{\stCn{i}}+\sum_{l=k+1}^n\ketbra{\stCn{l}}{\stCn{l}}\biggr)=k\I_n
\end{align*}
which is exactly equal to $(\id_n\tens\Tr_k)(\I)$. Consequently $(\id_n\tens\Tr_k)(T)=0$, i.e.~$T\in{S^\perp}$.

Similarly in the case $n\leq{k}$ we set
\[
\I+T=k\sum_{i,j=1}^n\ketbra{\bigstrut\stCn{i}}{\stCn{j}}\tens\ketbra{\bigstrut\stCk{i}}{\stCk{j}}
\]
To illustrate this choice take $n=3$ and $k=4$. Then, with the identification $\B(\CC^3\tens\CC^4)$ with $\Mat_4(\Mat_3(\CC))$ analogous to the one used above, $\I+T$ is
\[
\left[\begin{array}{ccc|ccc|ccc|ccc}
4 & 0 & 0 & 0 & 4 & 0 & 0 & 0 & 4 & 0 & 0 & 0 \\
0 & 0 & 0 & 0 & 0 & 0 & 0 & 0 & 0 & 0 & 0 & 0 \\
0 & 0 & 0 & 0 & 0 & 0 & 0 & 0 & 0 & 0 & 0 & 0 \\
\hline
0 & 0 & 0 & 0 & 0 & 0 & 0 & 0 & 0 & 0 & 0 & 0 \\
4 & 0 & 0 & 0 & 4 & 0 & 0 & 0 & 4 & 0 & 0 & 0 \\
0 & 0 & 0 & 0 & 0 & 0 & 0 & 0 & 0 & 0 & 0 & 0 \\
\hline
0 & 0 & 0 & 0 & 0 & 0 & 0 & 0 & 0 & 0 & 0 & 0 \\
0 & 0 & 0 & 0 & 0 & 0 & 0 & 0 & 0 & 0 & 0 & 0 \\
4 & 0 & 0 & 0 & 4 & 0 & 0 & 0 & 4 & 0 & 0 & 0 \\
\hline
0 & 0 & 0 & 0 & 0 & 0 & 0 & 0 & 0 & 0 & 0 & 0 \\
0 & 0 & 0 & 0 & 0 & 0 & 0 & 0 & 0 & 0 & 0 & 0 \\
0 & 0 & 0 & 0 & 0 & 0 & 0 & 0 & 0 & 0 & 0 & 0
\end{array}\right].
\]
Clearly, upon re-ordering of the basis $\I+T$ becomes
\[
\left[
\begin{array}{c|c}
0_{nk-n} & 0 \\
\hline
0 & kJ_n
\end{array}
\right]
\]
which is positive and of norm $k\|J_n\|=kn$. Finally, as in the previous case, we compute
\[
(\id_n\tens\Tr_k)(\I+T)=k\sum_{i=1}^n\ketbra{\stCn{i}}{\stCn{i}}=k\I_n=(\id_n\tens\Tr_k)(\I),
\]
so $(\id_n\tens\Tr_k)(T)=0$.
\end{proof}

Theorem \ref{thm:theta} gives us the following immediate corollary:

\begin{corollary}
We have $\tilde{\vartheta}(S)=k^2$.
\end{corollary}

\begin{proof}
Let us temporarily denote $S$ by $S_{n,k}$. Then for any $m$ we have
\[
\B(\CC^m)\tens{S_{n,k}}=\B(\CC^m)\tens\B(\CC^n)\tens\I_k=\B(\CC^m\tens\CC^n)\tens\I_k=S_{mn,k}.
\]
Thus by Theorem \ref{thm:theta}
\[
\vartheta\bigl(\B(\CC^m)\tens{S}\bigr)
=\begin{cases}
k^2&mn>k\\
mnk&mn\leq{k}
\end{cases}
\]
and it follows that $\tilde{\vartheta}(S)=\sup\limits_m\vartheta(B(\CC^n)\tens{S})=k^2$.
\end{proof}

\section*{Acknowledgment}

Research presented in this paper was partially supported by NCN (National Science Centre, Poland) grant no.~2022/47/B/ST1/00582.

\end{document}